\documentclass[10pt]{article}

\usepackage{amsmath, amsthm, amssymb}
\usepackage{geometry}
\usepackage{booktabs}
\usepackage{hyperref}
\usepackage{graphicx}
\usepackage{adjustbox}
\usepackage[numbers]{natbib}
\usepackage{rotating}
\usepackage{floatpag}

\theoremstyle{plain}
\newtheorem{theorem}{Theorem}
\newtheorem{lemma}[theorem]{Lemma}

\newtheorem{proposition}[theorem]{Proposition}

\theoremstyle{definition}
\newtheorem{definition}[theorem]{Definition}

\theoremstyle{remark}

\newcommand{\Z}{\mathbb{Z}}
\newcommand{\Aut}{\mathrm{Aut}}
\DeclareMathOperator{\lcm}{lcm}

\title{Enumerating Two-Orbit Graphs}
\author{David Seka and Stefan Szeider\\[4pt]
  \small Algorithms and Complexity Group\\[-3pt]
  \small TU Wien, Vienna, Austria\\[-3pt]
  \small \href{mailto:seka@ac.tuwien.ac.at}{\texttt{\char`\{seka,sz\char`\}@ac.tuwien.ac.at}}
}
\date{}

\begin{document}

\maketitle
\thispagestyle{empty}

\begin{abstract}
We present an approach to enumerate graphs whose automorphism group
has exactly two orbits.  Our method exploits the
observation that we can enumerate all graphs whose automorphism group
contains a given this permutation group.  We obtain the relevant
groups via Goursat's lemma.  In order to scale the enumeration, we
employ additional optimizations that prune irrelevant groups.
  In total, we enumerate, for the first time, all connected two-orbit graphs of up to $27$ vertices, totaling $10{,}094{,}721$ graphs, pushing the state of the art well beyond what direct enumeration methods can achieve.
\end{abstract}

\section{Introduction}

One-orbit graphs, or vertex-transitive graphs, are graphs whose automorphism group is transitive.
Many prominent graphs, such as the clique $K_n$, the complete bipartite graph $K_{n,n}$, or the Petersen graph, are one-orbit graphs.
A natural generalization is the class of two-orbit graphs, where the automorphism group has exactly two orbits on the vertices.
They are of interest as they tend to have large automorphism groups while being considerably more numerous.

Directly enumerating two-orbit graphs by generating all non-isomorphic graphs and testing their orbit structure is infeasible beyond very small sizes, since the total number of graphs grows super-exponentially.
However, by exploiting the group structure of the automorphism group, we reduce the search space dramatically, enabling enumeration up to $27$ vertices.

In this paper, we present a systematic approach to enumerating two-orbit graphs.
Ba\v{s}i\'{c} et al.~\cite{basic_fowler_pisanski} published counts of two-orbit graphs for certain orders, citing private communication with Royle.
However, neither the method used to obtain these counts nor a published list of the graphs appears to be available.

The closest related work is the census of transitive groups and vertex-transitive graphs by Holt and Royle~\cite{holt_royle}, which enumerates one-orbit graphs by working through a catalog of transitive permutation groups.
Our approach generalizes this strategy to two-orbit graphs by using Goursat's lemma to systematically generate candidate groups.
Orbital graph theory, which studies graphs whose edge sets are unions of orbitals of a permutation group, underpins our edge-orbit-based generation method.

\section{Alternative Approaches}

Before introducing our main approach, we describe two simpler methods we explored that motivated the current work.

The first is straightforward: enumerate all graphs up to isomorphism via nauty~\cite{nauty} and compute the number of orbits.
This approach is simple to implement and can be extended to more orbits.
Unfortunately, the total number of graphs explodes combinatorially.
Enumeration beyond $11$ vertices is computationally infeasible with this approach.

We therefore employed SMS~\cite{sms}, a SAT-based framework for graph enumeration that supports additional constraints on parameters, such as vertex degree.
For example, it is possible to constrain the degree of a vertex or even how many neighbors a vertex has among some subset of vertices.
We split the computation into multiple stages.
First, we fix the orbit sizes $k$ and $l$.
Then, we specify the \emph{orbit matrix}, which indicates how many neighbors each vertex in each orbit has among the vertices of that orbit.
Another parameter we fix is the number of triangles each vertex participates in within each orbit.
Furthermore, we detect whenever the subgraph induced on an orbit is fully determined and prune partial solutions where this is not the case.
These optimizations enabled enumeration of up to $16$ vertices, which required a parallel computational cluster to complete within a reasonable time.
These limitations motivated the group-theoretic approach described in Section~\ref{sec:groups}.

\section{From Groups to Graphs}
\label{sec:groups}
Enumerating the graphs that correspond to an automorphism group proceeds via edge orbits.

\begin{lemma}
  Let $G$ be a graph and $a \in \Aut(G)$.
  Then $\{v,w\} \in E(G)$ if and only if $\{a(v), a(w)\} \in E(G)$.
\end{lemma}
\begin{proof}
  Each element of $\Aut(G)$ permutes the vertices so that the adjacency matrix remains constant.
\end{proof}

\begin{definition}
  Let $H$ be a permutation group acting on a vertex set $V$.
  The \emph{edge orbits} of~$H$ are the orbits of~$H$ acting on the set $\binom{V}{2}$ of unordered vertex pairs.
\end{definition}

This means we can partition the set of vertex pairs into edge orbits by grouping pairs that lie in the same orbit under a given permutation group $H$.
Now, we generate all graphs invariant under $H$ by iterating over all subsets of edge orbits.
Indeed, this captures every graph whose automorphism group is a supergroup of $H$, as such a group produces a coarser partition of the edge set.
This method produces a large number of duplicate graphs, so we use nauty~\cite{nauty} to deduplicate them via their canonical form.

\begin{definition}
  A two-orbit permutation group $P$ acting on $\Omega = \Omega_1 \sqcup \Omega_2$ is called \emph{minimal} if it has no proper subgroup that is transitive on both $\Omega_1$ and $\Omega_2$.
\end{definition}

Since every two-orbit automorphism group contains a minimal two-orbit subgroup, and any graph invariant under the full automorphism group is also invariant under that minimal subgroup, it suffices to enumerate graphs arising from minimal two-orbit groups.

\subsection{Enumerating Relevant Groups}
While transitive permutation groups have been cataloged up to degree $47$~\cite{holt_royle}, the same cannot be said about two-orbit groups.
Such a group $C$ must be a subgroup of the direct product $A \times B$ of two transitive groups where the projection to either group is surjective.
We call a subgroup of the direct product of two groups that projects surjectively onto both factors a \emph{subdirect product}.
Goursat's Lemma characterizes such groups.

\begin{theorem}[Goursat's Lemma {\cite[p.~75]{lang}}]
Let $A, B$ be groups.
Each subdirect product $C$ can be identified by $(N_A, N_B, \phi)$, where $N_A \trianglelefteq A$, $N_B \trianglelefteq B$ and $\phi$ is an isomorphism between $A / N_A$ and $B / N_B$.
The subgroup is reconstructed as $C = \{(a,b) \in A \times B : \phi(aN_A) = bN_B\}$.
The normal subgroup $N_A$ corresponds to the kernel of the projection to $B$ and $N_B$ to the kernel of the projection to $A$.
\end{theorem}

The subgroup $C$ obtained in Goursat's Lemma is often viewed as the \emph{fiber product} of $A$ and $B$ over the common quotient $Q = A/N_A \cong B/N_B$. Concretely, $C$ consists of exactly those pairs $(a,b)$ whose images in $Q$ agree under the identification induced by $\phi$, thus $C$ synchronizes the actions of $A$ and $B$ through the same quotient. If $N_A = N_B = \{e\}$, then $Q \cong A \cong B$ and $C = \{(a,\phi(a)) : a \in A\}$ is the \emph{diagonal subgroup}, hence $C \cong A \cong B$. At the other extreme, if $N_A = A$ and $N_B = B$, then $Q$ is trivial, the compatibility condition is trivial, and therefore $C = A \times B$.

The automorphism group of a two-orbit graph must be transitive when projecting to one of the orbits.
Thus, it must be a subdirect product of two transitive groups $A,B$.
This leads to the following algorithm to enumerate graphs of size $n$:
For each partition $n = k + (n-k)$ with $1 \le k \le \lfloor n/2 \rfloor$, enumerate all transitive permutation groups of degrees~$k$ and $n-k$.
Then, for each pair of such groups, generate all subdirect products $C$ via Goursat's Lemma, and construct all graphs that are invariant under $C$.
This approach is complete, as any two-orbit graph $G$ is guaranteed to be generated when the graphs invariant under its automorphism group are enumerated.

\section{Optimizations}
When enumerating for larger values of $n$, the number of candidate groups grows rapidly and many groups are duplicates.
In practice, this creates a bottleneck, as the ever-growing runtime necessitates parallelization, leading to the following tradeoff: either write all groups to separate files and deduplicate afterward, or coordinate the tasks so that each deduplicates.
When testing larger instances, storage constraints on our compute cluster made the first option impractical, and coordinating the tasks would be technically challenging.

We therefore identify additional pruning strategies for parallel workers.
An effective strategy is to check for minimality.
Given a candidate two-orbit group $P$ with orbit sizes $k$ and $n-k$, we check minimality in stages.
First, if $|P| = \lcm(k, n-k)$, then $P$ is trivially minimal: any subgroup transitive on an orbit of size $m$ must have order divisible by $m$, so any two-orbit subgroup of $P$ has an order divisible by both $k$ and $n-k$, hence at least $\lcm(k, n-k) = |P|$; no proper subgroup can achieve this.
Otherwise, we search for a proper two-orbit subgroup by testing the derived subgroup, the minimal normal subgroups, and the maximal subgroups of $P$ in turn, checking whether each has order divisible by both $k$ and $n-k$ and acts transitively on both orbits.
If no such witness is found, $P$ is declared minimal.
Although this check incurs considerable computational cost, the resulting reduction is substantial: over $95\%$ of the subdirect products are found to be non-minimal for $n \leq 27$, yielding significantly fewer groups to expand in the later stages.

The necessary computations become expensive when both kernels are trivial, i.e.,\ $N_A = \{e\}$ and $N_B = \{e\}$.
In this case, $C$ is the diagonal subgroup introduced in Section~\ref{sec:groups}, acting simultaneously on both orbits.
Such a diagonal group is never minimal when both $A$ and $B$ are natural symmetric or alternating groups, specifically when $A$ is a natural $S_m$ with $m \geq 3$ or a natural $A_m$ with $m \geq 4$, and likewise for $B$.
In these cases $A$ has a proper transitive subgroup (e.g.\ $A_m < S_m$, or $V_4 < A_4$), whose diagonal copy is then a proper two-orbit subgroup.
We note that $S_2 \cong \Z_2$ and $A_3 \cong \Z_3$ are both minimally transitive, so the conditions $m \geq 3$ and $m \geq 4$ are necessary.

\subsection{Pruning Essential Kernels}
When only enumerating two-orbit groups that have no proper two-orbit subgroups, pruning criteria that discard certain subdirect products can improve the runtime considerably.
An effective approach is to only calculate the subdirect product for essential kernels.

\begin{proposition}
Let $A$ and $B$ be transitive permutation groups on $\Omega_1$ and $\Omega_2$, respectively, let $N_A \trianglelefteq A$ and $N_B \trianglelefteq B$, and let $\phi \colon A/N_A \to B/N_B$ be an isomorphism. Define $C=\{(a,b)\in A\times B : \phi(aN_A)=bN_B\}$.
We call $N_A$ an \emph{essential kernel} in $A$ if there is no proper transitive subgroup $M<A$ such that $N_A M=A$.
If $N_A$ is not essential, then $C$ is not minimal as a two-orbit permutation group on $\Omega_1 \sqcup \Omega_2$.
\end{proposition}

\begin{proof}
Suppose $M < A$ is a proper transitive subgroup with $N_A M = A$, and let $C_M := C \cap (M \times B)$.
Since $M < A$, $C_M$ is a proper subgroup of $C$.

We show $C_M$ projects surjectively onto both $M$ and $B$.
For surjectivity onto $M$: given any $m \in M$, choose $b \in B$ with $\phi(mN_A) = bN_B$; then $(m,b) \in C_M$.
For surjectivity onto~$B$: given any $b \in B$, since $C$ projects onto $B$ there exists $a \in A$ with $(a,b) \in C$. Write $a = nm$ with $n \in N_A$ and $m \in M$ (possible since $N_A M = A$). Then $aN_A = mN_A$, so $\phi(mN_A) = \phi(aN_A) = bN_B$, giving $(m,b) \in C_M$.

Since the action of $C_M$ on $\Omega_1 \sqcup \Omega_2$ factors through its projections, and $M$ is transitive on $\Omega_1$ and $B$ is transitive on $\Omega_2$, $C_M$ is transitive on each orbit.
Hence $C_M$ is a proper two-orbit subgroup of $C$, so $C$ is not minimal.
\end{proof}

This eliminates over $90\%$ of kernel pairs for $n \leq 27$.

\subsection{Deduplicating Orbital Configurations}
The parallel workers each produce a list of minimal two-orbit groups, and different workers can independently discover groups that give rise to the same graphs.
Rather than deduplicating at the graph level, we first deduplicate at the level of \emph{orbital configurations} (the partition of ordered vertex pairs into orbitals), since two groups with the same orbital configuration generate exactly the same set of graphs.
This deduplication is done in two rounds.

\paragraph{Round 1 (label-preserving).}
Two groups have the same orbital configuration if their orbital partitions are identical as sets of directed pair sets.
We compute a canonical signature for each group's orbitals (by sorting pairs within each orbital, then sorting the list of orbitals) and discard duplicates by signature.
This is a pure hash lookup that eliminates most duplicates.

\paragraph{Round 2 (scheme isomorphism).}
Two orbital configurations may generate the same graphs even if their signatures differ, namely when one configuration can be obtained from the other by simultaneously relabeling the vertices and the orbitals.
Formally, we check whether there exists a vertex permutation $\psi \colon V \to V$ and an orbital relabeling $\sigma$ such that
\[
  \mathcal{O}_2(\psi(i),\psi(j)) = \sigma(\mathcal{O}_1(i,j)) \quad \text{for all } i \neq j,
\]
where $\mathcal{O}_k(i,j)$ denotes the orbital index of a pair $(i,j)$ in configuration $k$.
We search for such a pair $(\psi, \sigma)$ by backtracking on the assignment of vertices, propagating implied~$\sigma$ constraints incrementally and pruning on contradiction.
After assigning two or three vertices, the constraints on $\sigma$ are typically fully determined, making the search tree very shallow in practice.

\section{Conclusion}
Applying the group-theoretic pipeline described above, we obtained, for the first time, a complete enumeration of all connected two-orbit graphs of up to $27$ vertices, yielding $10{,}094{,}721$ graphs in total.
The complete set of graphs is publicly available~\cite{zenodo}.
Table~\ref{tab:edge_orbits} provides the breakdown by number of edge orbits.
For orders where independent data is available, our results are confirmed by nauty (up to $n=11$), SMS (up to $n=16$), and the counts reported by Royle~\cite{basic_fowler_pisanski}.

\begin{table}[p]
\thisfloatpagestyle{empty}
\centering
\rotatebox{90}{%
\setlength{\tabcolsep}{2.5pt}
\small
\begin{tabular}{@{}rrrrrrrrrrrrrrrrrrrrrrrrrr@{}}
\toprule
  $k \setminus n$ &        3 &        4 &        5 &        6 &        7 &        8 &        9 &       10 &       11 &       12 &       13 &       14 &       15 &       16 &       17 &       18 &       19 &       20 &       21 &       22 &       23 &       24 &       25 &       26 &       27 \\
\midrule
       1 &        1 &        1 &        2 &        4 &        3 &        6 &        5 &        5 &        5 &        8 &        6 &        8 &       15 &       11 &        8 &       14 &        9 &       21 &       24 &       16 &       11 &       31 &       23 &       18 &       30 \\
       2 &        0 &        2 &        4 &       18 &       12 &       46 &       29 &       43 &       38 &       98 &       59 &      103 &      151 &      190 &      110 &      285 &      138 &      420 &      341 &      315 &      203 &      869 &      433 &      449 &      628 \\
       3 &        0 &        0 &        0 &       12 &        4 &       54 &       34 &       74 &       41 &      270 &       92 &      305 &      402 &      718 &      246 &     1341 &      378 &     2117 &     1332 &     1624 &      708 &     6279 &     1961 &     2968 &     3743 \\
       4 &        0 &        0 &        0 &        2 &        0 &       30 &       12 &       52 &       16 &      331 &       55 &      363 &      514 &     1352 &      266 &     2903 &      458 &     5318 &     2573 &     3621 &     1137 &    22524 &     4379 &     8593 &    10968 \\
       5 &        0 &        0 &        0 &        0 &        0 &        6 &        4 &       17 &        3 &      284 &       20 &      258 &      469 &     1738 &      203 &     4359 &      434 &     9211 &     3725 &     5488 &     1437 &    56544 &     7739 &    16838 &    23380 \\
       6 &        0 &        0 &        0 &        0 &        0 &        2 &        0 &        1 &        0 &      183 &        6 &      129 &      345 &     1879 &      118 &     5130 &      327 &    12453 &     4459 &     6462 &     1502 &   112054 &    11823 &    26114 &    41411 \\
       7 &        0 &        0 &        0 &        0 &        0 &        0 &        0 &        0 &        0 &      110 &        0 &       50 &      251 &     1831 &       60 &     5496 &      196 &    14313 &     4999 &     6614 &     1316 &   190905 &    16078 &    35084 &    66769 \\
       8 &        0 &        0 &        0 &        0 &        0 &        0 &        0 &        0 &        0 &       53 &        0 &       13 &      152 &     1787 &       23 &     5305 &       97 &    14885 &     5255 &     6056 &      969 &   292831 &    19694 &    42043 &    99346 \\
       9 &        0 &        0 &        0 &        0 &        0 &        0 &        0 &        0 &        0 &       22 &        0 &        2 &       87 &     1627 &        7 &     4714 &       41 &    14377 &     5316 &     4993 &      647 &   416618 &    22044 &    45677 &   139401 \\
      10 &        0 &        0 &        0 &        0 &        0 &        0 &        0 &        0 &        0 &        3 &        0 &        0 &       30 &     1427 &        1 &     3597 &       13 &    13039 &     4845 &     3683 &      382 &   555666 &    22452 &    45324 &   182434 \\
      11 &        0 &        0 &        0 &        0 &        0 &        0 &        0 &        0 &        0 &        0 &        0 &        0 &        8 &     1086 &        0 &     2365 &        2 &    11191 &     3992 &     2411 &      191 &   694869 &    20502 &    40992 &   221082 \\
      12 &        0 &        0 &        0 &        0 &        0 &        0 &        0 &        0 &        0 &        0 &        0 &        0 &        0 &      734 &        0 &     1213 &        0 &     9054 &     2763 &     1406 &       76 &   809588 &    16446 &    33697 &   242741 \\
      13 &        0 &        0 &        0 &        0 &        0 &        0 &        0 &        0 &        0 &        0 &        0 &        0 &        0 &      392 &        0 &      499 &        0 &     6826 &     1626 &      721 &       20 &   872753 &    11346 &    25145 &   239567 \\
      14 &        0 &        0 &        0 &        0 &        0 &        0 &        0 &        0 &        0 &        0 &        0 &        0 &        0 &      169 &        0 &      128 &        0 &     4666 &      747 &      318 &        2 &   863949 &     6566 &    16956 &   208850 \\
      15 &        0 &        0 &        0 &        0 &        0 &        0 &        0 &        0 &        0 &        0 &        0 &        0 &        0 &       49 &        0 &       21 &        0 &     2832 &      277 &      117 &        0 &   780210 &     3118 &    10260 &   160119 \\
      16 &        0 &        0 &        0 &        0 &        0 &        0 &        0 &        0 &        0 &        0 &        0 &        0 &        0 &        9 &        0 &        0 &        0 &     1457 &       66 &       30 &        0 &   639027 &     1179 &     5482 &   106253 \\
      17 &        0 &        0 &        0 &        0 &        0 &        0 &        0 &        0 &        0 &        0 &        0 &        0 &        0 &        0 &        0 &        0 &        0 &      624 &       12 &        5 &        0 &   471486 &      340 &     2524 &    60839 \\
      18 &        0 &        0 &        0 &        0 &        0 &        0 &        0 &        0 &        0 &        0 &        0 &        0 &        0 &        0 &        0 &        0 &        0 &      204 &        0 &        0 &        0 &   311318 &       69 &      959 &    29385 \\
      19 &        0 &        0 &        0 &        0 &        0 &        0 &        0 &        0 &        0 &        0 &        0 &        0 &        0 &        0 &        0 &        0 &        0 &       48 &        0 &        0 &        0 &   182116 &        8 &      288 &    11915 \\
      20 &        0 &        0 &        0 &        0 &        0 &        0 &        0 &        0 &        0 &        0 &        0 &        0 &        0 &        0 &        0 &        0 &        0 &        6 &        0 &        0 &        0 &    93435 &        0 &       60 &     3857 \\
      21 &        0 &        0 &        0 &        0 &        0 &        0 &        0 &        0 &        0 &        0 &        0 &        0 &        0 &        0 &        0 &        0 &        0 &        0 &        0 &        0 &        0 &    41330 &        0 &        8 &      993 \\
      22 &        0 &        0 &        0 &        0 &        0 &        0 &        0 &        0 &        0 &        0 &        0 &        0 &        0 &        0 &        0 &        0 &        0 &        0 &        0 &        0 &        0 &    15463 &        0 &        0 &      173 \\
      23 &        0 &        0 &        0 &        0 &        0 &        0 &        0 &        0 &        0 &        0 &        0 &        0 &        0 &        0 &        0 &        0 &        0 &        0 &        0 &        0 &        0 &     4716 &        0 &        0 &       21 \\
      24 &        0 &        0 &        0 &        0 &        0 &        0 &        0 &        0 &        0 &        0 &        0 &        0 &        0 &        0 &        0 &        0 &        0 &        0 &        0 &        0 &        0 &     1120 &        0 &        0 &        0 \\
      25 &        0 &        0 &        0 &        0 &        0 &        0 &        0 &        0 &        0 &        0 &        0 &        0 &        0 &        0 &        0 &        0 &        0 &        0 &        0 &        0 &        0 &      180 &        0 &        0 &        0 \\
      26 &        0 &        0 &        0 &        0 &        0 &        0 &        0 &        0 &        0 &        0 &        0 &        0 &        0 &        0 &        0 &        0 &        0 &        0 &        0 &        0 &        0 &       14 &        0 &        0 &        0 \\
\midrule
  \textbf{total} &        1 &        3 &        6 &       36 &       19 &      144 &       84 &      192 &      103 &     1362 &      238 &     1231 &     2424 &    14999 &     1042 &    37370 &     2093 &   123062 &    42352 &    43880 &     8601 &  7435895 &   166200 &   359479 &  1853905 \\
\bottomrule
\end{tabular}
}
\caption{Number of connected two-orbit graphs by number of edge orbits ($k$) and graph size ($n$)}
\label{tab:edge_orbits}
\end{table}
 
The computation of the groups and the enumeration dominate the runtime, but both steps admit further optimization.
For the group enumeration, it may be beneficial to omit the minimality computation when essential kernel pruning already eliminates most groups.
Furthermore, additional lightweight criteria certifying non-minimality may exist.

The most promising avenue for further optimization is graph generation.
For graphs with $30$ orbitals, for instance, we must generate and canonically label more than a billion graphs, the vast majority of which are duplicates.
One may exploit isomorphisms of orbital-respecting edge colorings to prune the search tree.

\end{document}